\documentclass[12pt,reqno]{amsart}

\usepackage{amsmath}
\usepackage{amssymb}

\numberwithin{equation}{section}

\newcommand{\C}{{\mathbb{C}}}

\newcommand{\HH}{{\mathbb{H}}}

\newcommand{\R}{{\mathbb{R}}}

\newcommand{\T}{{\mathbb{T}}}
\newcommand{\Z}{{\mathbb{Z}}}

\newcommand{\Aa}{{\mathcal{A}}}

\newcommand{\Cc}{{\mathcal{C}}}

\newcommand{\Gg}{{\mathcal{G}}}
\newcommand{\Hh}{{\mathcal{H}}}

\newcommand{\Ll}{{\mathcal{L}}}
\newcommand{\Mm}{{\mathcal{M}}}

\newcommand{\swann}{\mathcal{U}(N)}

\newcommand{\euler}{\mathcal{X}_{0}}
\newcommand{\imag}{\mathrm{\mathbf{i}}}

\newcommand{\ms}[1]{\mathsf{#1}}

\newcommand{\mf}[1]{\mathfrak{#1}}
\newcommand{\mbb}[1]{\mathbb{#1}}

\newcommand{\abs}[1]{\left\lvert #1 \right\rvert}
\newcommand{\pair}[1]{\left\langle #1 \right\rangle}
\newcommand{\set}[1]{\left\lbrace #1 \right\rbrace}

\newcommand{\eqst}[1]{\begin{equation*} #1 
                      \end{equation*}}
\newcommand{\eq}[1]{\begin{equation} #1
                    \end{equation}}
\newcommand{\alst}[1]{\begin{align*} #1
                      \end{align*}}
\newcommand{\al}[1]{\begin{align} #1
                      \end{align}}

\theoremstyle{plain}
\newtheorem{thm}{Theorem}[section]
\newtheorem{lem}[thm]{Lemma}
\newtheorem*{lemma*}{Lemma}

\theoremstyle{definition}

\theoremstyle{definition}
\newtheorem{defn}{Definition}

\newtheorem{ex}{Example}

\theoremstyle{remark}
\newtheorem*{question*}{Question}

\DeclareMathOperator{\End}{End}

\DeclareMathOperator{\Map}{Map}
\DeclareMathOperator{\Hom}{Hom}

\DeclareMathOperator{\grad}{grad}

\usepackage{hyperref}
\hypersetup{
   colorlinks = true,
    linkcolor = blue,
    citecolor = blue,   
}

\allowdisplaybreaks

\begin{document}

\title[Generalised monopole equations on K\"ahler surfaces]{Generalised monopole equations on K\"ahler surfaces}

\author[I. Biswas]{Indranil Biswas}

\address{School of Mathematics, Tata Institute of Fundamental
Research, Homi Bhabha Road, Mumbai 400005, India}
\email{indranil@math.tifr.res.in}

\author[V. Thakre]{Varun Thakre}

\address{International Centre for Theoretical Sciences (ICTS-TIFR), Hesaraghatta, Hobli, Bengaluru 560089, India}

\email{varun.thakre@icts.res.in}


\date{Revised on \today }


\begin{abstract}

In this article, we establish a Hitchin-Kobayashi type correspondence for generalised Seiberg-Witten monopole equations on K\"ahler surfaces. We show that the ``stability'' criterion we obtain, for the existence of solutions, coincides with that of the usual Seiberg-Witten monopole equations. This enables us to construct a map from the moduli space of solutions to the generalised equations to effective divisors.
\end{abstract}

\maketitle

\section{Introduction}

In this article, we study a generalisation of the Seiberg-Witten (SW) monopole equations on a K\"ahler surface. Let $(X, g_X)$ be a smooth, oriented, four-dimensional Riemannian manifold. Fix a ${\rm Spin}^c$-structure $Q \longrightarrow X$. Spinor bundles are vector bundles associated to $Q$, with respect to a certain standard action on the vector space of quaternions $\HH$. The idea behind the generalisation is to replace the spinor representation $\HH$ with a hyperK\"ahler manifold $(M, g_{M}, I_1, I_2, I_3)$ admitting certain symmetries. \emph{Generalised spinors} are the sections of the associated fiber bundle. It is then possible to construct a non-linear Dirac operator, acting on the sections of the fiber-bundle. The operator is a first order, non-linear elliptic operator. This is the essence of the generalisation of Seiberg-Witten (GSW) monopole equations. An appropriate replacement of the quadratic map, which maps spinors to self-dual 2-forms on $X$, gives the GSW monopole equations. The generalisation was first introduced by C. Taubes \cite{taubes} in three dimensions. It was extended to four dimensions by V. Pidstrygach \cite{victor}. However, such generalisations of the Dirac operator were already known to physicists and have been used in the study of gauged $\sigma$-sigma models \cite{anselmi-fre95}, \cite{bagger-witten83}.

On a K\"ahler surface, the generalised monopole equations were studied by R. Waldm\"uller \cite{waldmuller03} and K. Strokorb \cite{strokorb09} in their Diploma thesis and by A. Haydys \cite{andriy} in his Ph.D 
thesis. The equations reduce to a system of \emph{twisted}, symplectic vortex 
equations (see \cite{andriy}, Sec. 4.2). The latter are a system of vortex-like equations with values in a symplectic manifold $(F, \omega)$ and can be defined over any compact K\"ahler manifold. The equations were discovered independently by I. Mundet i Riera \cite{riera99} and K. Cieliebak, A.R. Gaio
and D. Salamon \cite{ciel-giao-salamon}. I. Mundet i Riera obtained a Hitchin-Kobayashi-type correspondence when $F$ is K\"ahler. The correspondence relates the spaces of solutions up to real and complex gauge transformation and coincides with the notion of ``stability'' which arises in the construction of algebraic moduli space, by using Geometric Invariant Theory (GIT).

The aim of this article is to explicitly evaluate the stability condition for the 
(Abelian) GSW equations on a K\"ahler surface, for a large class of hyperK\"ahler manifolds, admitting a hyperK\"ahler potential. We show 
that the condition can be reduced to the existence and the uniqueness of 
solutions to Kazdan-Warner equation. This, however, coincides with condition for the existence of solutions to the usual SW monopole equations.

Given this, it is tempting to ask if \emph{there exists a map between 
moduli space of gauge-equivalent solutions to GSW and effective divisors on $X$?} 
Section \ref{sec: hk correspondance} provides an affirmative answer to this
question.

\section{HyperK\"ahler manifolds}
\label{preliminaries}

Let $(M, g_{M}, I_1, I_2, I_3)$ be a $4n$-dimensional hyperK\"ahler manifold. Let ${\rm Sp}(1)$ denote the group of unit quaternions and $\mf{sp}(1)$ its Lie algebra. As a matter of convenience, we think of the complex structures as covariantly constant endomorphisms of $TM$ with values in $\mf{sp}(1)^{\ast} = \left(\mf{Im}(\HH)\right)^{\ast}$
 \eq{
 \label{eq:algebra homomorphism}
 I \in \Gamma(M, \End(TM) \otimes \mf{sp}(1)^{\ast}), ~~~ I_\xi := \xi_1 I_1 + \xi_2 I_2 + \xi_3 I_3, ~~ \xi \in \mf{sp}(1).
 }
 
\noindent It is easy to see that $M$ has an entire family of K\"ahler structures
parametrized by $S^2 \subset \mf{Im}(\HH)$. Let $\omega_i, ~ i=1,2,3$, denote the K\"ahler 2-forms associated to $I_1, I_2, I_3$. Combining the three K\"ahler 2-forms, we define a single, $\mf{sp}(1)$-valued 2-form
\eqst{
\omega := i\omega_1 + j\omega_2 + k\omega_3.
}

Suppose that a Lie group $G$ acts isometrically on $M$. Let $\mf{g}$ denote its Lie algebra. We will denote by 
\eqst{
K^M: \mf{g} \longrightarrow \Gamma (M, TM), ~~~ \gamma \longmapsto K^M_{\gamma}
} 
the Killing vector field on $M$ due to $\gamma$.

\begin{defn}
\label{def: permuting action}
An isometric (left) action of ${\rm Sp}(1)$ on $M$ is said to be \emph{permuting} if
\eqst{
(L_q)^{\ast} \omega \,=\, \overline{q} \, \omega \, q\, , \ \ q \,\in\, {\rm Sp}(1)\, .
}
\end{defn} 

\noindent In other words, the induced action of ${\rm Sp}(1)$, on the two-sphere of complex 
structures, is the standard action of ${\rm Sp}(1)$ on $S^2$.

\begin{defn}
\label{def: tri-holomorphic action}
An isometric action of a Lie group $G$ on $M$ 
is \emph{tri-holomorphic} (or \emph{hyperK\"ahler}), if it fixes the 2-sphere of complex structures
\eqst{
\Ll_{K^M_{\eta}} \, \omega = 0, ~~~ \eta \in \mf{g}.
}

\noindent If, in addition, the $G$ action is Hamiltonian with respect to each $\omega_i$, then
the action is said to be \emph{tri-Hamiltonian} (or hyperHamiltonian). We can define a
single $G$-equivariant \emph{hyperK\"ahler moment map}
\eq{
\mu\,:\, M \,\longrightarrow\, \mf{sp}(1)^{\ast}\otimes {\rm Lie}(G)^{\ast}\,=\,
\mf{sp}(1)^{\ast}\otimes \mf{g}^{\ast}\, ,
}
by combining the three moment maps into one
\eqst{
\mu \,=\, i\mu_1 + j\mu_2 + k\mu_3\, .
}
\end{defn}

Amongst the class of hyperK\"ahler manifolds, which admit a permuting ${\rm Sp}(1)$-action, there
are those that also admit a \emph{hyperK\"ahler potential}; i.e, a smooth map
$\rho_0: M \longrightarrow \R^{+}$, which is simultaneously a K\"ahler potential for
each $\omega_i$. Swann \cite{swann} shows that for such hyperK\"ahler manifolds,
the permuting ${\rm Sp}(1)$ action can be extended to a homothetic action of $\HH^{\ast}
\,=\,{\rm Sp}(1) \times \R^{+}$ and the gradient vector field $\euler = \grad(\rho_0) = -I_{\xi}K^M_{\xi}$ is independent of $\xi \in \mf{sp}(1)$.
Moreover, $\euler$ generates the \emph{homothetic} action of $\R^{+} \subset \HH^{\ast}$: 
\eqst{
(L_{r})^{\ast}g_{M}(\cdot, \cdot) \,=\, r^2g_{M}(\cdot, \cdot)\, .
}

\begin{defn}
A quaternionic-K\"ahler manifold is a $4n$-dimensional Riemannian manifold
whose holonomy is contained in ${\rm Sp}(n){\rm Sp}(1)\,:=\, ({\rm Sp}(n)\times {\rm Sp}(1))/\pm 1$.
\end{defn}

\begin{thm}[{\cite{swann}}]\label{thm: swanns theorem}
Let $M$ be a hyperK\"ahler manifold admitting a hyperK\"ahler potential $\rho_0$. Then
$\rho_0^{-1}(c)/{\rm Sp}(1)\,:=\, N$ is a quaternionic-K\"ahler manifold
of positive scalar curvature.
\end{thm}

On the other hand, starting with a quaternionic-K\"ahler manifold $N$ of positive scalar curvature, Swann's construction produces a hyperK\"ahler manifold $\swann$ with a permuting ${\rm Sp}(1)$-action and a hyperK\"ahler potential. This is the total space of the fiber bundle $\swann \,\longrightarrow\, N$ with a typical fiber $\HH^{\ast}/(\Z/2\Z)$. Moreover, Swann shows that any action of a Lie group $G$ on $N$ which preserves the quaternionic-K\"ahler structure, can be lifted to a tri-Hamiltonian action of $G$ on $\swann$. In this case, the moment map has a simple expression (see Sec. 3.3 of \cite{henrik}):
\eq{
\label{eq: moment map on swann bundle}
\pair{\mu, \xi\otimes\eta} = -\frac{1}{2}g_M(K^M_{\xi}, K^M_{\eta}), ~~ \xi \in \mf{sp}(1)~ \text{and}~ \eta \in \mf{g}.
}

Examples for compact, quaternionic-K\"ahler manifolds, with positive scalar curvature are given by \emph{Wolf spaces}. These are compact, homogeneous, quaternionic-K\"ahler manifolds classified by Wolf \cite{wolf65} and Alekseevskii \cite{alexeevski68}. The list includes quaternionic projective spaces $\HH P^n = \frac{{\rm Sp}(n+1)}{{\rm Sp}(n)
\times {\rm Sp}(1)}$, complex Grassmannians $X^n \,= \,\frac{{\rm SU}(n)}{{\rm S}({\rm U}(n-
2) \times {\rm U}(2))}$, real Grassmannians $Y^n \,=\, \frac{{\rm SO}(n)}{{\rm SO}(n-4)
\times {\rm SO}(4))}$, etc. The associated manifolds $\swann$ are certain co-adjoint
orbits of complex simple Lie groups (see \cite{swann}).

\subsection{Target hyperK\"ahler manifold}

Let $(M, g_{M}, I_1, I_2, I_3)$ be a $4n$-dimensional hyperK\"ahler manifold. Suppose
that there is an isometric action of ${\rm U}(1)$ on $M$, that preserves $\omega_1$ and rotates $\omega_2$ and $\omega_3$; in other words, if $X$ is the Killing vector field
on $M$ that generates the action, then
\eqst{
\Ll_X \omega_1 = 0, ~~ \Ll_X \omega_2 = -\omega_3, ~~ \Ll_X\omega_3 = \omega_2.
}

Such an action is called a \emph{rotating action} of ${\rm U}(1)$. This notion was  introduced by N. Hitchin, A. Karlhede, U. Lindstr\"om and M. Ro\v{c}ek \cite{hitchin87} . Henceforth, we will refer to such a hyperK\"ahler manifold 
as a \emph{target hyperK\"ahler manifold}.

\begin{ex}
\label{ex: rotating S1 action}
Consider the flat quaternionic space $\HH^n$. If we write $\HH^n = \C^n \oplus j\C^n$, we have
\eqst{
\omega_1 \,=\, \frac{\imag}{2}\sum^n_{l=1} dz_l \wedge d\overline{z_l} + dw_l \wedge
d\overline{w_l}, ~~ \omega_c \,:=\, \omega_2 + i\omega_3
\,=\, \sum^n_{l=1} dz_l \wedge dw_l\, .
}
Then the circle action $\big(e^{i\theta}, \,(z,w)\big) \,\longmapsto\, (z,\, e^{i\theta}\cdot w)$ is a rotating action, fixing $\omega_1$.
\end{ex}

\begin{ex}
\label{ex: rotating S1 action swann bundles}

Given $\swann$ for some $N$ of positive scalar curvature, the stabilizer ${\rm U}(1)_r
\,\subset\, {\rm Sp}(1)$ of $I_1$ gives the requisite rotating action.
\end{ex}

\section{Generalised Seiberg-Witten equations on K\"ahler surface}

Fix a target hyperK\"ahler manifold $M$ and assume that there is a tri-Hamiltonian 
action of ${\rm U}(1)$ on $M$ that commutes with the rotating ${\rm U}(1)_r$-action. To 
distinguish this group action from the rotating one, we denote this group by 
${\rm U}(1)_0$. Therefore, $M$ has a \emph{rotating action} of $\T^2 \,=\, {\rm U}(1)_r
\times_{\Z/2\Z} {\rm U}(1)_0$.

Let $X$ be a K\"ahler surface, and let $\omega_X$ be the K\"ahler 2-form. The 
K\"ahler structure on $X$ determines the reduction of its ${\rm SO}(4)$-frame bundle 
to a principal ${\rm U}(2)$-bundle $P_{{\rm U}(2)}$. More precisely, ${\rm U}(2) \,=\,
({\rm U}(1)_r \times {\rm Sp}(1)_-)/(\Z/2\Z)$, where ${\rm U}(1)_r \,\subset \,{\rm 
Sp}(1)_+$ is the stabilizer of the complex structure $R_{\overline{i}}$ in the ${\rm 
SO}(4) \,=\, ({\rm Sp}(1)_+\times {\rm Sp}(1)_-)/(\Z/2\Z)$-representation $\R^4 
\,\cong\, \HH$. The ${\rm U}(1)$-bundle $P_r \,:=\, P_{{\rm U}(2)}/{\rm Sp}(1)_-$ is 
precisely the one associated to the anti-canonical line bundle of $X$. Fix an 
auxiliary principal ${\rm U}(1)$-bundle $P_0$ over $X$ and define the $\T^2$-bundle 
$P_{\mbb{T}^2}\,:=\, (P_{r} \times_X P_0)/(\Z/2\Z)$.

\emph{Generalised spinors} are $\T^2$-equivariant maps 
\eqst{
\Map(P_{\T^2}, M)^{\T^2} \cong \Gamma(X, \ms{M}), ~~~~~ \text{where} ~~~~~ \ms{M} = P_{\T^2} \times_{\scriptscriptstyle \T^2} M.
} 

The Levi-Civita connection on $X$ defines a connection on 
$P_{{\rm U}(2)}$. Therefore, a connection $A$ on $P_0$ and the Levi-Civita connection 
together determine a unique connection $\ms{A}$ on $P_{\T^2}$. A spinor $u\,:\, P_{\T^2} 
\,\longrightarrow \,M$ and the connection $\ms{A}$ determine a $\T^2$-equivariant map
$K^M_{\ms{A}}|_u\,:\, TP_{\T^2} \,\longrightarrow\, TM$ as follows. For any $v \,\in\, T_p P_{\T^2}$,
take $$K^M_{\ms{A}}(v)|_u \,:=\, K^M_{\ms{A}(v)}|_{u(p)} \,\in\, T_{u(p)}M$$
(note that $\ms{A}(v)\, \in\, \mf{t}^2 :=\text{Lie}(\mbb{T}^2)$).

The differential of $u$ is also $\T^2$-equivariant. We define the covariant derivative of $u \in \Map(P_{\T^2}, M)^{\T^2}$ with respect to $\ms{A}$ to be the one-form $D_{\ms{A}}u \in \Omega^1(P_{\T^2}, u^{\ast}TM)^{\T^2}$
\eqst{
D_{\ms{A}}u \,=\, du + K^M_{\ms{A}}|_u\, .
}
This is an equivariant, horizontal one-form on $P_{\T^2}$. Indeed, for any $\xi \in \mf{t}^2$,
we have
\eqst{
D_{\ms{A}}u \left(K^{P_{\T^2}}_{\xi} \right) = du\left(K^{P_{\T^2}}_{\xi} \right)\, + \, K^M_{\ms{A}\left(K^{P_{\T^2}}_{\xi} \right)}|_u = - K^M_{\xi}|_u \, + \, K^M_{\xi}|_u = 0.
}
Therefore, $D_{\ms{A}}u$ descends to a one-form on $X$ with values in $(u^{\ast}TM)/\T^2$. Denote by
$\overline{\partial}_{\ms{A}}u$, the $(0,1)$-part of this 1-form, meaning
\eqst{
\overline{\partial}_{\ms{A}}u \,=\, \frac{1}{2} \left(D_{\ms{A}}u \, - \, I_1 \circ D_{\ms{A}}u \circ \widetilde{I}_X \right),
}
where $\widetilde{I}_X$ is the lift of the complex structure $I_X$ to the horizontal subspace $\Hh_{\ms{A}} \subset TP_{\T^2}$.

\noindent Note that in defining the $\overline{\partial}_{\ms{A}}$ operator, we treat $M$ as a K\"ahler manifold with respect to complex structure $-I_1$.

Denote by $\Aa(P_0)$ the space of connections on $P_0$. Define the \emph{configuration space }
\eq{
\label{eq: configuration space}
\Cc \,:= \, \Map\,(P_{\T^2}, \, M)^{\T^2} \times \Aa(P_0).
}
Let $\Gg \, := \, \Map \, (X, {\rm U}(1))$ be the infinite-dimensional \emph{gauge group}. Then, the configuration space carries a (right) action of $\Gg$.

\begin{thm}[{\cite{andriy}}]
\label{them: gen sw on kahler surface}
Let $(X,\,\omega_X)$ be a K\"ahler surface. Then for a pair $(u, \,A) \,\in\, \Cc$, the
perturbed, GSW equations on $X$ reduce to the following system:
\eq{\label{eq: gen sw on kahler surface}
\left\{
    \begin{array}{lcl}
     \overline{\partial}_{\ms{A}}u = 0 \\      
     \Lambda_{\omega_X} F_{A} \, + \, \imag\mu_1 \circ u \, + \, \imag t = 0, ~~ t\in R \\
      \mu_c \circ u = 0, ~~~ F^{0,2}_{A} = 0
    \end{array}
  \right.}
where $F_{A}$ is the curvature of $A$, and $\mu_c$ is the complex moment map
$\mu_2 + \imag \mu_3$. Moreover, these equations are invariant under the action of $\Gg$.
\end{thm}

\section{A Hitchin-Kobayashi-type correspondence}
\label{sec: hk correspondance}

In this section, we establish a Hitchin-Kobayashi-type correspondence for the solutions of \eqref{eq: gen sw on kahler surface}. To understand what we mean by this, observe that the configuration space is naturally a K\"ahler manifold, with respect to the complex structure induced by the complex structures $I_X$ on $X$ and $I_1$ on $M$ (see \cite{riera99}, Sec. 2.3). The (right) action of $\Gg$ on $\Cc$ extends in a natural way to the action of its complexification $\Gg^{\C}\, = \, \Map \, (X, \, \C^{\ast})$, with respect to the induced complex structure.

The first and third equations of \eqref{eq: gen sw on kahler surface} are invariant under the action of $\Gg^{\C}$ whereas the second equation is invariant only under the action of $\Gg$. One would like to know if and when there exists a transformation $g \in \Gg^{\C}$, such that
\eqst{
\Lambda_{\omega_X}F_{g\cdot A} + \imag\mu_1\circ (g\cdot u) + \imag t = 0.
}
The necessary and sufficient conditions for the existence of such a gauge transformation lies at the heart of Hitchin-Kobayashi correspondence. The condition coincides with the notion of stability that arises in the algebraic construction of the moduli space of solutions to various gauge-theoretic equations using Geometric Invariant Theory (GIT).

One can also view this from the point of view of Kempf-Ness theory in infinite dimensions. The action of the gauge group $\Gg$ on $\Cc$ is holomorphic with the associated infinite dimensional moment map given by
\eqst{
\Upsilon_t(u,A)\,:=\, i\Lambda_{\omega_X}F_A + \imag\mu_1\circ u + \imag t\, .
}

\noindent Let $\Aa^{1,1}(P_0) \subset \Aa(P_0)$ be the space connections on $P_0$, whose curvature is of the form $(1,1)$. Then $\Cc^{1,1}:= \Map(P_{\T^2}, \, \swann)^{\T^2} \times \Aa^{1,1}(P_0)$ is a complex subvariety of $\Cc$, with a holomorphic action of $\Gg$, that extends to an action of $\Gg^{\C}$. The statement of correspondence now narrows down to asking when does a $\Gg^{\C}$-orbit in $\Cc^{1,1}$ intersect the zero of the infinite-dimensional moment map. This is a common paradigm in gauge theory and was pioneered by M. Atiyah and R. Bott \cite{atiyah-bott82}. It has been used in several other contexts, most notably by Donaldson \cite{donaldson83}, \cite{donaldson85} and by Uhlenbeck and Yau \cite{uhlenbeck-yau86} to relate stable vector bundles over complex manifolds with Hermitian-Einstein vector bundles. The idea has been subsequently used in the study of various other gauge theoretic equations \cite{oscar94}, \cite{oscar-bradlow95}, \cite{teleman96}. However, since we are interested in existence of solutions to \eqref{eq: gen sw on kahler surface}, we will additionally demand that the first equation of \eqref{eq: gen sw on kahler surface} also be satisfied.

A more general criterion of stability has been obtained by Mundet i Riera \cite{riera99} for K\"ahler vortex equations over compact K\"ahler manifolds. These are a system of vortex-like equations with values in a K\"ahler manifold $(F, \omega)$ and can be defined over any compact
K\"ahler manifold. In general, this criterion is not easy to evaluate. However, for target hyperK\"ahler manifold with a hyperK\"ahler potential, we show that the condition reduces to Kazdan-Warner equations. The existence of a hyperK\"ahler potential lies at the heart of this computation.

To effect our computations below, we need the completion of the configuration space and the gauge group in an appropriate $(k,p)$-Sobolev norm. We will assume that the Sobolev exponent $k-\frac{4}{p}>0$. It is in this setting that we evaluate the stability criterion. The assumption is implicit in our computations that follow. For details on Sobolev completion of maps between manifolds, we refer the reader to Subsection 4.1, Appendix B of \cite{wehrheim}. 

In order to give a clear picture of our construction, we begin by considering the simplest 
possible generalisation.

\subsection{A simpler case: Seiberg-Witten with multiple spinors}

Let $M$ be the flat quaternionic space $\Hom_{\C}(\C^n,\, \HH)$, where $\HH$ is regarded
as a complex vector space with respect to the complex structure
$R_{\overline{i}}$. The standard complex volume form on $\C^n$ defines a complex linear
isomorphism $(\C^n)^{\ast} \,\cong\, \C^n$. We can therefore identify $M \,\cong\, \HH
\otimes_{\C} \C^n \,\cong\, \HH^n$. The natural action of the group
${ \rm SU}(n) \times {\rm U}(1)_0$ on $\HH\otimes_{\C}\C^n$ corresponds to an action
$${\rm SU}(n) \times {\rm U}(1)_0 \,\hookrightarrow \,{\rm U}(n) \,\hookrightarrow 
\,{\rm Sp}(n) \,\curvearrowright\, \HH^n\, .$$ For $\alpha,\, \beta \,\in\, \C^n$, the moment
map associated to the ${\rm U}(1)_0$-action is given by
\eqst{
\mu_1(\alpha + \beta j) \,=\, \displaystyle -\sum^n_{i=1}\frac{\abs{\alpha_i}^2 -
\abs{\beta_i}^2}{2}, ~~ \mu_c(\alpha + \beta j)
\,=\, \displaystyle -\sum^n_{i=1}\pair{\alpha_i,\beta_i}\, ,
}
where $\alpha + \beta j \,=\, \sum^n_{i=1} \alpha_i s_i + j \sum^n_{i=1}\beta_i s_{n+i}$, and
$\{s_i\}$ is the complex spinor basis.

The rotating action of ${\rm U}(1)_r \subset {\rm Sp}(1)_+$ on $\HH\otimes_{\C}\C^n$ is
\eqst{
z \cdot (\alpha,\, \beta) \,\longmapsto\, (\alpha,\, \beta\, \overline{z})\, .
}

Fix a principal ${\rm SU}(n)$ bundle $P_{{\rm SU}(n)}$ over $X$, and denote by $Q$ the principal
bundle $P_{\T^2} \times_X P_{{\rm SU}(n)}$. As a representation of $\T^2 \times {\rm SU}(n)$, note
that $M$ decomposes as 
\eqst{
M \,=\, \C^n \oplus W\otimes_{\C} \C^n\, ,
}
where $W$ is the representation of ${\rm U}(1)_r$ on $\Lambda^{0,2}(\R^4)^{\ast} \,\cong
\, \C$. In particular,  the associated fiber bundle $P_{{\rm U}(1)_r} \times_{{\rm U}(1)_r} W$ is
the anti-canonical line bundle $K^{-1}_X$ over $X$. Therefore, any spinor
$$u \,\in \,\Map(Q, \, M)^{\T^2 \times {\rm SU}(n)}$$ decomposes into two components $f$ and $g$. In terms of the complex basis $\{s_i \}$, we can write a spinor $u = \displaystyle\sum^n_{i=1}f_i \, s_i + \displaystyle\sum^{n}_{i=1}g_i \, s_{i}$.

Assume that there exists a connection $B$ on $P_{{\rm SU}(n)}$ compatible with
the holomorphic structure, which means that
$F_B^{0,2} \,=\, 0$. Fix such a connection $B$. The configuration space is given by
\eqst{
\Cc \,=\, \Map \, (Q, \, M)^{\T^2 \times {\rm SU}(n)} \times \Aa(P_0)\, .
}

For a pair $(u,\, A) \,\in\, \Cc$, the (perturbed) SW equations with multiple spinors, on a compact, K\"ahler surface $X$ are
\eq{\label{eq: sw with multiple spinors on kahler surface}
\left\{
    \begin{array}{lcl}
     \displaystyle \sum^n_{i=1}\overline{\partial}_{A\otimes B}\, f_i + \overline{\partial}^{\ast}_{A\otimes B} \, g_i = 0 \\      
     \Lambda_{\omega_X} F_{A} \, - \, \imag \left( \displaystyle\sum^n_{i=1}\frac{\abs{f_i}^2 - \abs{g_i}^2}{2} \right) \, + \, \imag t = 0, ~~ t\in R \\
      \displaystyle\sum^n_{i=1}\pair{f_i,g_i} = 0, ~~
      F^{0,2}_{A} = 0
    \end{array}
  \right.}  
where $\pair{\cdot, \cdot}$ is the standard Hermitian inner product on $\C^n$. 

The K\"ahler structure on the configuration space is induced by the complex structures $I_X$ on $X$ and $R_{\overline{i}}$ on $\HH^n$. The moment map for the holomorphic action of $\Gg$ on $\Cc$ is
\eqst{
\Upsilon_t(u,A) \,=\, \Lambda_{\omega_X} F_{A} - \imag \left( \displaystyle\sum^n_{i=1}\frac{\abs{f_i}^2 - \abs{g_i}^2}{2} \right) + \imag t. 
}
Define $\Hh^{1,1} \subset \Cc$ to be the complex subvariety 
\eqst{
\Hh^{1,1} \, = \, \set{(u, A)\in \Cc ~|~ \overline{\partial}_{A\otimes B} \, u = 0 ~ \text{and} ~ F_A^{0,2} = 0}.
}
The moduli space of solutions to \eqref{eq: sw with multiple spinors on kahler surface} is a K\"ahler submanifold of $\Upsilon_t^{-1}(0)/\Gg$, given by
\eqst{
\Mm (B, g_X) \,:=\, \left(\Hh^{1,1} \cap \Upsilon_t^{-1}(0)\cap \set{(u,A)
\,\in\, \Cc ~|~\displaystyle\sum^n_{i=1}\pair{f_i,g_i} = 0}\right)/\Gg
}

Let $\Hh^{ss} \,=\, \set{(u,A)\in \Hh^{1,1}~|~ (f_1, f_2, \cdots \cdot , f_n) \not\equiv 0~\text{and}~\displaystyle\sum^n_{i=1}\pair{f_i, \, g_i} = 0} \subset \Hh^{1,1}$.

\begin{thm}[{\cite{bryan-wentworth96}}]
\label{thm: multiple spinors bryan and wenthworth}
The moduli space of solutions to \eqref{eq: sw with multiple spinors on kahler surface}
has a holomorphic description
\eqst{
\Mm (B, g_{X}) \,\cong\, \Hh^{ss}/\Gg^{\C}\, .
}
\end{thm}

The above correspondence reduces to a Kazdan-Warner type equation, which gives the necessary condition for existence of solutions to be
\eq{
\label{eq: existence for sw with multiple spinors}
t \,>\, \frac{4\pi}{vol(X)}\deg_{\omega_X} P_0\, .
}
Recall that the degree of the bundle $P_0$ is given by
\eqst{
\deg_{\omega_X} P_0  = \frac{\imag}{2\pi}\int_X \Lambda_{\omega_X}F_A. 
}

A choice of a large enough $t$ ensures that $(f_1, \,f_2,\, \cdots ,\, f_n) \,\not\equiv\, 0$. Bryan and Wentworth obtained the above correspondence when $P_{{\rm SU}(n)}$ is trivial and $B$ is a trivial connection. However, a verbatim argument carries over to the case when $P_{{\rm SU}(n)}$ is non-trivial.

The statement of Theorem \ref{thm: multiple spinors bryan and wenthworth} is an infinite dimensional analogue of a finite-dimensional principal. Namely, suppose that we are given a smooth, projective variety $W$, with a holomorphic action of a reductive Lie group $G^{\C}$. Let $\mu$ denote the moment map for the action of $G$ on $W$. Then the algebraic and the symplectic quotients agree; i.e,
\eqst{
\mu^{-1}(0)/G \cong W^{ss}/G^{\C}
}
where $W^{ss}$ is a dense open set.

Under the assumption \eqref{eq: existence for sw with multiple spinors} we will now establish a map from $\Mm(B, g_X)$ to the moduli space of solutions to the usual SW monopole equations 
\eq{\label{eq: sw on kahler surface}
\left\{
    \begin{array}{lcl}
     \overline{\partial}_{A}\alpha \,=\, 0, ~~~ \beta = 0 \\      
     \Lambda_{\omega_X} F_{A} - \imag \left(\frac{\abs{\alpha}^2}{2} - t \right)= 0, ~~ t\in R \\
    \pair{\alpha,\beta} = 0\\[1mm]
      F^{0,2}_{A} = 0
    \end{array}
  \right.}
on $X$. We denote the latter moduli space by $\Mm^{SW}(g_X)$.

\subsubsection{\textbf{SW with multiple spinors $\Rightarrow$ SW }}

We will denote by $u$ the equivariant map $\Map(Q, \, \HH^n)^{\T^2\times {\rm SU}(n)}$ and by $\phi$ a positive spinor; i.e, a $\T^2$-equivariant map $\Map(P_{\T^2}, \,\HH)^{\T^2}$.

Let $(u,A)$ be a solution to \eqref{eq: sw with multiple spinors on kahler surface} and suppose that there exists a $\phi = \alpha + \beta$ such that 
\eqst{
\frac{\imag}{2} \left(\abs{\alpha}^2 - \abs{\beta}^2 \right) = \imag \left( \displaystyle\sum^n_{i=1}\frac{\abs{f_i}^2 - \abs{g_i}^2}{2} \right), ~~~ \pair{\alpha,\, \beta} = 0.
}
Here $\phi = \alpha + \beta$ is the usual decomposition of the spinor on a K\"ahler surface. Owing to \eqref{eq: existence for sw with multiple spinors}, a non-trivial $\phi$, satisfying the above equations, always exists. Moreover, the condition \eqref{eq: existence for sw with multiple spinors} also implies that any solution $\phi$ to the monopole equations will have $\beta = 0$. Therefore, without loss of generality, we may assume that $\beta = 0$. Therefore, pair $(\phi, A)$ satisfies
\eq{
\label{eq: part sw on kahler surface}
\left\{
    \begin{array}{lcl}
    \Lambda_{\omega_X} F_{A} - \frac{\imag}{2} \left(\abs{\alpha}^2  - t \right) = 0 \\[1mm]
      F^{0,2}_{A} = 0.
    \end{array}
  \right.}
  
\begin{lem}
\label{lem: sw with mult. spinor implies hol. spinor}
The spinor $\phi$ satisfying \eqref{eq: part sw on kahler surface} is holomorphic; i.e.,
\eqst{
\overline{\partial}_{A}\alpha = 0.
}
\end{lem}

\begin{proof}
We have
\eqst{
d\left(\frac{\abs{\alpha}^2}{2}\right)
= d\left( \displaystyle\sum^n_{i=1}\frac{\abs{f_i}^2 - \abs{g_i}^2}{2} \right).
}
Computing the left hand side:
\alst{
\frac{\pair{\alpha, D_{A}\alpha}_{\R}}{2}
& = \frac{1}{2}\left (\pair{\alpha, \, \overline{\partial}_{A}\alpha}_{\R} + \pair{\alpha, \, \partial_{A}\alpha}_{\R} \right)
}
where $\pair{\cdot, \cdot}_{\R}$ denotes the real part of the respective Hermitian inner products. Similarly, on the right hand side we have
\alst{
d\left( \displaystyle\sum^n_{i=1}\frac{\abs{f_i}^2 - \abs{g_i}^2}{2} \right)
& = \displaystyle\sum^n_{i=1} \frac{1}{2}\left(\left( \pair{f_i, \, \overline{\partial}_{A\otimes B} \, f_i}_{\R} \right) + \left(\pair{f_i, \, \partial_{A\otimes B} \, f_i}_{\R} - \pair{g_i, \,  \partial_{A\otimes B} \, g_i}_{\R} \right) \right)
}
Equating the $(0,1)$-parts on both the sides, we get
$\displaystyle \pair{\alpha, \overline{\partial}_{A}\alpha}_{\R} = \displaystyle\sum^n_{i=1} \pair{f_i,\, \overline{\partial}_{A\otimes B}\,f_i}_{\R}.$ The equation $\displaystyle\sum^n_{i=1} \pair{f_i, g_i} = 0$ implies
\eqst{
d^{\ast}\displaystyle\sum^n_{i=1} \pair{f_i,  \, g_i} = \displaystyle\sum^n_{i=1} \pair{f_i,  \, \overline{\partial}^{\ast}_{A\otimes B}  \, g_i} = 0 \implies \displaystyle\sum^n_{i=1}\pair{f_i,  \, \overline{\partial}^{\ast}_{A\otimes B}  \, g_i}_{\R}
\,=\, 0\, .
}
Together, this gives
\eqst{
\pair{\alpha, \overline{\partial}_{A}\alpha}_{\R} \,=\, \displaystyle\sum^n_{i=1} \pair{f_i,  \, \overline{\partial}_{A\otimes B} \, f_i + \overline{\partial}^{\ast}_{A\otimes B}  \, g_i}_{\R} = \pair{\left(\displaystyle\sum^n_{i=1} f_i \right),\left(\displaystyle\sum^n_{i=1} \overline{\partial}_{A\otimes B} \, f_i + \overline{\partial}^{\ast}_{A\otimes B} \,  g_i \right)}_{\R}.
}
The last equality follows from the fact that $\pair{f_i, \, \overline{\partial}_{A\otimes B}f_j} = \pair{f_i, \, \overline{\partial}^{\ast}_{A\otimes B}  \, g_j} = 0$ for $i \neq j$. Therefore, from \eqref{eq: sw with multiple spinors on kahler surface}, we conclude that $\pair{\alpha,  \, \overline{\partial}_{A} \, \alpha}_{\R} = 0$ and so, the $(0,2)$-form
\eqst{
\overline{\partial}\pair{\alpha, \overline{\partial}_{A}\alpha} \,=\, \pair{\overline{\partial}_{A}\alpha\wedge\overline{\partial}_{A}\alpha} \,=\, 0\, .
}
The statement of the Lemma follows.
\end{proof}

In particular, we have shown that $(A, \phi)$ is a solution to the monopole equations \eqref{eq: sw on kahler surface}. The uniqueness of the solution $(\phi, A)$ is easily seen.
  
\subsection{General case: Swann bundles}

We will now implement the above program in a more general setting where $\HH^n$ is replaced by a more general target hyperK\"ahler manifold, with a hyperK\"ahler potential; i.e., the total space of a Swann bundle. The strategy for the program is the same as that for multi-monopole equations discussed in the previous subsection. 

Define $\Hh^{1,1} \,=\, \set{(u,A) \in \Cc ~|~ \overline{\partial}_{\ms{A}}u = 0 ~ \text{and} ~ F_A^{0,2} = 0}$. Then, the moduli space of solutions to \eqref{eq: gen sw on kahler surface} is once again a K\"ahler submanifold of $\Upsilon_t^{-1}(0)/\Gg$, given by
\eqst{
\mathfrak{M} = \left(\Hh^{1,1} \cap \Upsilon_t^{-1}(0)\cap \set{u \in \Map(P_{\T^2},\, \swann)^{\T^2} ~|~\mu_c\circ u = 0}\right)/\Gg.
}

Let $F_0 \subset \swann$ be the fixed-point set of the~ $U(1)_0$ action on $\swann$. Consider the dense open subset of $\Hh^{1,1}$ 
\eqst{
\Hh^{ss} := \set{(u,A) \in \Hh^{1,1}~|~ u(P_{\T^2}) \not\subset F_0}.
}

\begin{thm}[\textbf{Hitchin-Kobayashi correspondence}]
Let $(u,A) \in \Hh^{1,1}$ and assume that $t > \frac{4\pi}{vol(X)}\deg_{\omega_X} P_0$. Then, the moduli space $\mathfrak{M}$ is non-empty. Moreover, $\mathfrak{M}$ has a holomorphic description
\eqst{
\Mm \cong \Hh^{ss}/\Gg^{\C}.
}
\end{thm}
\begin{proof}
Our aim is to find conditions under which there exists a $g \in \Gg^{\C}$ such that 
\eqst{
\Upsilon_t(g\cdot u, g \cdot A) = 0.
}
 Consider an element $e^f \in \Gg^{\C}$. If $f$ is purely imaginary, then $e^f \in \Gg$. Since the equations are invariant under $\Gg$, we consider the case when $f$ is real. Now the complexified gauge group $\Gg^{\C}$ acts on $\Hh^{1,1}$ as
\eqst{
e^f \cdot A \longmapsto A + \partial \overline{f} - \overline{\partial} f ~~~ \text{and} ~~~ e^f \cdot u, ~~ \text{for}~ e^f \in \Gg^{\C}.
}
So $F_{e^f \cdot A} \,=\, F_A + \overline{\partial}\partial f - \partial\overline{\partial} f$. From \eqref{eq: moment map on swann bundle}, the moment map component $\mu_1\circ u$ can be written down explicitly as
\eqst{
\mu_1 \circ u = - \frac{1}{2}g_M(K^M_{\xi_1}|_u, K^M_{\mathrm{i}}|_u)}
where $\xi_1 \in S^2 \subset \mf{sp}(1)$ is the basis element fixed by the rotating $U(1)_r$-action. Owing to the homothetic $\R^{\ast}$-action on $\swann$, we have
\eqst{
\mu_1 \circ (e^f \cdot u) = - e^{2f} \frac{1}{2}\cdot g_M(K^M_{\xi_1}|_u, K^M_{\mathrm{i}}|_u).
}
Observe that $\frac{1}{2} g_M(K^M_{\xi_1}|_u,\, K^M_{\mathrm{i}}|_u)
\,:\, P_{\T^2} \,\longrightarrow\, \R$ is $\T^2$-invariant and so we can think of it as a smooth, real-valued function on $X$. For simplicity, let $\mf{a}(u) = \frac{1}{2} g_M(K^M_{\xi_1}|_u, K^M_{\mathrm{i}}|_u)$. We can therefore write
\al{
\label{eq: kazdan-warner}
\Lambda_{\omega_X}F_{e^f\cdot A} + \imag\mu_1\circ (e^f\cdot u) \nonumber + \imag t
& = \Lambda_{\omega_X}(\overline{\partial}\partial - \partial\overline{\partial})f - \frac{\imag}{2} e^{2f}\mf{a}(u) + \Lambda_{\omega_X}F_{A} + \imag t. \nonumber
}
Hence, in order to find a $g \in \Gg^{\C}$ such that $\Upsilon_t(g\cdot u, g \cdot A) = 0$, we need to solve
\eq{
\label{eq: gsw kazdan-warner}
\Delta_X f + e^{2f}\mf{a}(u) \,=\, (t - 2\imag\Lambda_{\omega_X}F_{A})
}
where $\Delta_X$ is the positive definite Laplacian on $X$. Let $w \,=\,
t - 2\imag\Lambda_{\omega_X}F_{A}$. 

We now recall a result of \cite{kazdan-warner74} which will be used.

\begin{lem}[{\textbf{Kazdan-Warner}}]
\label{lem: kazdan-warner}
Let $X$ be a compact Riemannian manifold, and  let $B$ and $w$ be smooth functions on $X$ with $B$ being positive outside of a measure zero set and $\int_X w \,>\, 0$. Let $\Delta_X \,=\, -2\imag \partial \overline{\partial}$ be the negative definite Laplacian on $X$. Then the equation 
\eqst{
\Delta_X f +B(x) e^{2f} - w \,=\, 0
}
has a unique solution.
\end{lem}

The condition $\int_X w \,>\, 0$ translates to fixing a $ t \,>\, \frac{4\pi}{vol(X)}\deg_{\omega_X} P_0$. It follows that there exists a unique solution to \eqref{eq: gsw kazdan-warner}.The statement of the theorem follows. 
\end{proof}

A technical requirement in Lemma \ref{lem: kazdan-warner} is that the function $B$ be a positive function, outside of a measure zero set. A priori, it is unclear why this should hold for an abstract map $\mf{a}(u)$. However, in the following section, we will show that solutions to \eqref{eq: gen sw on kahler surface} determine a unique solution to \eqref{eq: sw on kahler surface}. This in turn will imply that $\mf{a}(u) = \abs{\alpha}^2$. Therefore, the technicality is automatically satisfied for $ t \,>\, \frac{4\pi}{vol(X)}\deg_{\omega_X} P_0$.

\subsubsection{\textbf{Solutions to GSW $\Rightarrow$ SW}}

Assume that $t > \frac{4\pi}{vol(X)}\deg_{\omega_X} P_0$ and let $(u,A)$ be a solution to \eqref{eq: gen sw on kahler surface}. Moreover, let $\phi$ be a usual spinor, satisfying
\eqst{
\left\{
    \begin{array}{lcl}
      \Lambda_{\omega_X} F_{A} - \imag \left(\frac{\abs{\alpha}^2 }{2} - t \right) = 0, ~~ t\in R \\
    F^{0,2}_{A} = 0
    \end{array}
  \right.}
In particular, $\mu_1\circ u = \frac{\abs{\alpha}^2 }{2}$. Once again, owing to the fact that $t > \frac{4\pi}{vol(X)}\deg_{\omega_X} P_0$, such a spinor always exists and and has $\beta = 0$, where $\phi = \alpha + \beta$. The second condition is already satisfied since $(u, A)$ is a solution to \eqref{eq: gen sw on kahler surface}.

\begin{thm}
\label{thm: gen. sw implies hol. spinor}
A spinor $\phi$ satisfying the above equation is holomorphic; i.e.,
\eqst{
\overline{\partial}_{A}\alpha \,=\, 0\, .
}
\end{thm}
\begin{proof}
Observe that
\eqst{
d\left(\frac{\abs{\alpha}^2}{2}\right) = d (\mf{a}(u)) = D_{\ms{A}} (\mu_1\circ u) = d\mu_1(u) (D_{\ms{A}}u).
}
We can split the 1-forms on left-hand side and right-hand side into its $(0,1)$ and $(1,0)$-components
and equate them to get
\eqst{
\pair{\alpha, \overline{\partial}_{A}\alpha} \,=\,
 d\mu_1(u) (\overline{\partial}_{\ms{A}}u) \,=\, 0\, .
}
Then, arguing as before, we have $\abs{\overline{\partial}_{A}\alpha}^2 \,=\, 0$. The statement of the theorem follows. 
Therefore, $(\phi, A)$ is a solution to SW monopole equations \eqref{eq: sw on kahler surface}.  
\end{proof}

\subsection{From SW to GSW}

It is possible to prove the converse of Theorem \ref{thm: gen. sw implies hol. spinor} and 
Lemma \ref{lem: sw with mult. spinor implies hol. spinor}. In other words, starting with a 
solution to vortex equations, it is possible to construct a solution to generalised equations.

Assume that $t \,>\, \frac{4\pi}{vol(X)} \deg_{\omega_X} P_0$, so that there exists a pair
$(\phi,\, A)$ satisfying \eqref{eq: sw on kahler surface}. Fix $M$ to be either $\HH^n$ or
$\swann$ for some $N$. The given condition then implies that there must exist a generalised spinor $u$, such that 
\begin{equation}\label{eq: inf. moment map gsw}
\Lambda_{\omega_X} F_A + \imag \mu_1\circ u + \imag t \,=\, 0\, .
\end{equation}

To show that $(u,\, A)$ is a solution to \eqref{eq: gen sw on kahler surface}, we must show that 
$\overline{\partial}_A u = 0$. From \eqref{eq: sw on kahler surface}, we know that
$\overline{\partial}_A \alpha = 0$. Since $(u, A)$ satisfies \eqref{eq: inf. moment map gsw} 
we also know that $\mu_1\circ u = \frac{\abs{\alpha}^2}{2}$. In particular,
we have $d(\mu_1\circ u) \,=\, 
d\left(\frac{\abs{\alpha}^2}{2}\right)$. Equating the $(1,0)$ parts on both sides we get that
$d\mu_1 \left(\overline{\partial}_A u \right) \,=\, 0$. If $\overline{\partial}_A u$ is not 
identically zero, then $\overline{\partial}_A u (p) \in \ker d\mu_1(u(p)) \subset T_{u(p)} M$ for 
every $p \in T_p P_0$, which in turn implies that $\mu_1\circ u (p) = 0$ for every $p \in 
P_0$. In particular we have $\mu_1\circ u = 0$. But this is a contradiction since $\alpha \neq 0$.
It must therefore be the case that $\overline{\partial}_A u = 0$. In conclusion, $(u, A)$ is a 
solution to \eqref{eq: gen sw on kahler surface}.

\section{Maps between moduli spaces}

In both the cases discussed above, over a K\"ahler surface, we get an explicit description of the map from the moduli space of solutions to the generalised equations to that of the usual SW monopole equations. More precisely,
\al{
\label{eq: map of moduli spaces}
& \Pi: \mathfrak{M}(g_X, M) \longrightarrow \Mm^{SW}(g_X), ~~ [(u,A)] \longmapsto [(\phi, A)], ~ \text{where} ~ \mu_1\circ u = \frac{\abs{\alpha}^2}{2}.
}
The fiber of the map is the set of all solutions $u$, which are holomorphic with respect to $\ms{A}$ and $\mu_1\circ u = \frac{\abs{\alpha}^2}{2}$. Since the solutions to SW are in one-to-one correspondence with effective divisors, $\Pi$ maps a solution to the generalised SW equations to an effective divisor $D$, given by the zeroes of the function $\mf{a}(u)$. 

A more general version of the correspondence \eqref{eq: map of moduli spaces} was studied by the second author in \cite{varun17}. Namely, the following theorem was proved:
\begin{thm}
\label{thm: corr. gen. sw and sw}
On a compact Riemannian manifold $X$, suppose that there exists a solution of the GSW equations. The composition $\mu\circ u$ defines a self-dual 2-form on $X$, which we denote by $\Omega$. Then, away from the set of degenerate points of $\Omega$, the equations \eqref{eq: gen sw on kahler surface} can be expressed as a second order PDE in terms of $\Omega$:
\begin{multline}
\label{eq: full ODE of Donaldson}
\nabla^{\ast}\nabla \Omega = - \left(\frac{s}{2} + \abs{\Omega}^2\right) \Omega - 2 \langle d\Omega + \ast d\abs{\Omega}, N_{\Omega} \rangle + \frac{1}{2} \left(\frac{|d\Omega|^2}{\abs{\Omega}^2} - |N_{\Omega}|^2 \right) \Omega \\
+ \frac{1}{2}\left(|d\abs{\Omega}|^2 + 2\langle d\abs{\Omega}, \ast d\Omega \rangle\right) \frac{\Omega}{\abs{\Omega}^2}
\end{multline}
\end{thm}

This is a generalisation of Donaldson's result \cite{donaldson}, who showed that the solutions to the usual SW equations are in one-to-one correspondence with self-dual 2-forms satisfying \eqref{eq: full ODE of Donaldson}, away from the singular set. It follows that there is a map between the moduli spaces of solutions to the GSW and SW equations. On a K\"ahler surface, the equation \eqref{eq: full ODE of Donaldson} reduces to a PDE for a single function, which is reminiscent of the formulation of two-dimensional vortex equations by Jaffe and Taubes \cite{jaffe-taubes82}. In the K\"ahler situation, \eqref{eq: full ODE of Donaldson} gives an alternate description of the map $\Pi$.

\bibliographystyle{ieeetr}

\end{document}